\documentclass[conference]{IEEEtran}
\usepackage{cite}
\usepackage{amsmath,amssymb,amsfonts}
\usepackage{algorithmic}
\usepackage{graphicx,color}
\usepackage{textcomp}
\usepackage{indentfirst}
\usepackage{cbc}
\usepackage{xcolor, soul}
\colorlet{soulcyan}{cyan!40}
\sethlcolor{soulcyan}

\usepackage{physics}
\usepackage[hidelinks]{hyperref}
\usepackage{bm}
\usepackage{amsthm}
\usepackage{float}

\theoremstyle{definition}
\newtheorem{definition}{Definition}
\newtheorem{theorem}{Theorem}
\newtheorem{lemma}{Lemma}
\newtheorem{remark}{Remark}
\newtheorem{corollary}{Corollary}

\def\BibTeX{{\rm B\kern-.05em{\sc i\kern-.025em b}\kern-.08em
    T\kern-.1667em\lower.7ex\hbox{E}\kern-.125emX}}
\AtBeginDocument{\definecolor{ojcolor}{cmyk}{0.93,0.59,0.15,0.02}}

\begin{document}

\title{RIS-assisted Cell-Free MIMO with Dynamic Arrivals and Departures of Users: A Novel Network Stability Approach}

\author{Charbel Bou Chaaaya$^\dag$, Mohamad Assaad$^\ddag$, and Tijani Chahed$^\mathsection$\\ $^\dag$ Centre for Wireless Communications, University of Oulu, Finland   \\$^\ddag$ Universit\'e Paris-Saclay, CNRS, CentraleSup\'elec, Laboratoire des Signaux et Syst\`emes, France\\
$^\mathsection$  Institut Polytechnique de Paris, Télécom SudParis, 19 Place Marguerite Perey, 91120 Palaiseau, France}
\maketitle

\begin{abstract}
Reconfigurable Intelligent Surfaces (RIS) have recently emerged as a hot research topic, being widely advocated as a candidate technology for next generation wireless communications. These surfaces passively alter the behavior of propagation environments enhancing the performance of wireless communication systems. In this paper, we study the use of RIS in cell-free multiple-input multiple-output (MIMO) setting where distributed service antennas, called Access Points (APs), simultaneously serve the users in the network. While most existing works focus on the physical layer improvements RIS carry, less attention has been paid to the impact of dynamic arrivals and departures of the users on the system performance. In such a case, ensuring the stability of the network is the main goal. For that, we propose an optimization framework of the phase shifts, for which we derive a low-complexity solution. We rigorously prove that our low  complexity solution stabilizes a guaranteed fraction (higher than 78.5\%) of the stability region, and all this stability region can be achieved by the optimal solution of our framework. We also provide  numerical results that corroborate the theoretical claims.
\end{abstract}



\section{INTRODUCTION}
\IEEEPARstart{U}{nprecedented} communication requirements for Beyond 5G networks will call for essentially new schemes. Numerous technologies have been investigated in the last few decades to satisfy the exponential surge for wireless connectivity, notably massive Multiple-Input Multiple-Output (mMIMO)\cite{marzetta2010noncooperative,hajri2016scheduling}, ultra dense networks (UDN) and millimeter waves (mmWave) \cite{boccardi2014five}. The reliable service offered by these innovations however, comes at the cost of expensive hardware and high energy consumption. Consequently, extensive research is ongoing to find novel designs for sustainable wireless networks with spectrum and power efficiencies, and low hardware cost.
Reconfigurable intelligent surfaces (RIS), or intelligent reflective surfaces (IRS), have recently emerged as a paradigm that can leverage engineered scattering surfaces to transmit and receive information \cite{liaskos2018new}. By smartly tuning the phase shift of each element, the reflected signals from different paths can be coherently combined at the desired receiver to improve the quality of the received signal. Accordingly, the RIS can maneuver the propagation environment to boost coverage over the area of interest and avoid blockages.  

\par Although incorporating RIS in wireless systems entails challenging optimization problems, its usage has shown real potential in enhancing point-to-point communications, as well as uplink and downlink multi-user schemes \cite{wu2019intelligent}. Within this framework, recent works on RIS-aided wireless systems have predominantly studied the physical layer advantages, while the protocol aspects for their integration have been generally overlooked. In addition, less attention has been paid to control signaling and network stability, defined as the set of users arrival rates that the network can serve within finite time, which are imperative for their operation. For instance, the authors in \cite{croisfelt2022random} conceived a simple random access algorithm for RIS-assisted communications, and \cite{cao2022massive} presented a next generation multiple access (NGMA) scheme with RIS. Moreover, \cite{yang2021reconfigurable} tackled the problem of non-orthogonal multiple access (NOMA) with RIS, and \cite{yang2020energy} proposed a resource allocation scheme with rate splitting multiple access (RSMA). On the other hand, \cite{9665300} studied a RIS-aided cell-free MIMO system where distributed access points (APs) in the coverage area cooperate via a backhaul unit to serve the users. This architecture alleviates the intercell interference that characterizes mMIMO \cite{denis2021improving}, and will be used as the physical layer in this work. Furthermore, \cite{cellfreeris2021} considerd a RIS-assisted cell free massive MIMO scheme with random phase shifts at the RIS and conjugate beamforming at the massive base station (BS) antennas. A closed-form expression of the achievable rate is obtained in \cite{cellfreeris2021} and a gain in terms of system coverage and user rate is shown. The aforementioned works focused on the physical layer without considering the impact of dynamic traffic on system performance. 

Contrary  to the previous studies that consider a fixed set of users in the system, we focus in this paper on a more realistic scenario where users dynamically arrive, exchange bursts of information and then leave the network once they are served. The design of phase shifts in this case must take into account the traffic/flow level, where flows represent file transfer.   

\par In the context of dynamic arrivals and departures of users, the network stability is usually a main metric to consider in the analysis and algorithmic design of the network  \cite{assaad2018power}. Network stability implies that all files/users will be served in a finite time. Network stability with dynamic population has been considered in \cite{assaad2018power} to deal with power control in massive MIMO. In this paper, we consider a different problem of phase shifts in RIS-assisted cell-free MIMO. A main question is this case is to know what is the optimization problem that the system has to solve (to find the phase-shifts at each time) in order to ensure that the network stays stable. While existing works focused on maximizing throughput and/or other related functions, we propose a phase shift framework  that depends on the current flows intensity and 
Effective Signal to Interference and Noise Ratio ($\bar{ESINR}$) of the users (as we will explain in Section III in the paper). Interestingly, the proposed framework can stabilize the network whenever it is possible (i.e. allows  achieving max stability region). We provide a low complexity solution of this optimization problem and prove that it achieves a high fraction of the max stability region (higher than 78.5\%). To the best of our knowledge, this is the first work that  considers a RIS-assisted wireless environment with dynamic arrivals and departures of the users.




\section{SYSTEM MODEL} \label{sec:2}
\subsection{CHANNEL MODEL}

\par We consider the uplink of a RIS-aided cell-free MIMO system, as illustrated in \autoref{fig:system}, where $N$ access points (APs) equipped with a single antenna each, serve $K$ single antenna users who transmit their signals simultaneously on the same time-frequency resource. All APs are linked to a central processing unit (CPU) via a backhaul network. The communication is assisted by an RIS comprising of $M$ reflecting elements. Typically, the number of RIS engineered elements $M$ is very large. We focus on the scenario where the direct channels between the users and the APs are weak as compared to the ones reflected by the RIS.  This is typically the case when  the RIS is installed for example in an indoor environment (such as a mall), in a university campus, or in an outdoor area with many buildings and obstacles. The RIS can then modify  the phases of incident signals and the phase shifts are adjusted adaptively by a controller linked to the CPU.
We consider a channel model similar to the one adopted in \cite{bjornson2020rayleigh,9665300}. We assume a block fading model where the channel realizations are generated randomly and are independent between blocks. The channels between user $k=1,\dots,K$ and the RIS, and the RIS and AP $n=1,\dots,N$ are respectively denoted by $\bm{g}_k \in \C^{M \times 1}$ and $\bm{h}_n \in \C^{M \times 1}$, and are modeled as spatially correlated ergodic processes  $\bm{g}_k \sim \mathcal{CN}\left(\bm{0}_{M \times 1},\, \overline{\bm{R}}_k\right)$ and $\bm{h}_n \sim \mathcal{CN}\left(\bm{0}_{M \times 1},\, \widetilde{\bm{R}}_n\right)$, where $\overline{\bm{R}}_k, \widetilde{\bm{R}}_n \in \C^{M \times M}$ are deterministic positive semi-definite covariance matrices that describe the spatial correlation between the channels of the RIS elements. Further, we can write $\overline{\bm{R}}_k = \overline{\alpha}_k \bm{R}_\text{\fontfamily{cmr}\selectfont r}$ and $\widetilde{\bm{R}}_n = \widetilde{\alpha}_n \bm{R}_\text{\fontfamily{cmr}\selectfont t}$, with $\overline{\alpha}_k$, $\widetilde{\alpha}_n \in \C$ being the large-scale fading coefficients, and $\bm{R}_\text{\fontfamily{cmr}\selectfont r}$, $\bm{R}_\text{\fontfamily{cmr}\selectfont t}$ being the RIS receive and transmit correlation matrices respectively. These matrices can follow a general model that depends on the size, distance and layout of RIS scattering elements, and the propagation environment. For example, a popular model for such planar arrays is the Kroenecker model where the correlation matrix is Hermitian Toeplitz with exponential entries $\rho^{\abs{j-i}}$ at row $i$ and column $j$, if $j \geq i$, where $\rho \in \C$ is the correlation coefficient that satisfies $\abs{\rho} \leq 1$. For the case of isotropic scattering in front of the RIS, the correlation matrix can be explicitly obtained from \cite[Proposition 1]{bjornson2020rayleigh}. We assume for the rest of this paper, that the transmit and receive correlation matrices are not equal, and the only technical assumption is that their diagonal elements are equal to unity due to power normalization. 

\par The cascaded channel between user $k$ and AP $n$ is then written:
\begin{equation}
    u_{n,k} = \bm{h}_n^{H} \, \bm{\Theta} \, \bm{g}_k
\end{equation}
where $\bm{\Theta} = diag(\bm{\phi})$ is the RIS reflection matrix, with $\bm{\phi} = \left[\phi_1, \dots, \phi_M\right]^\intercal$, $\phi_m = e^{j\theta_m}$ and $\bm{\theta} = \left[\theta_1, \dots, \theta_M\right]^\intercal \in \left[0, 2\pi\right]^M$ are phase shifts of the $M$ elements. We define the aggregated channel from user $k$ to the APs as $\bm{u}_k = [u_{1,k}, \dots, u_{N,k}]^\intercal$, where $u_{n,k} = \sum_{m=1}^M h_{nm}\, g_{km}\, e^{j\theta_m}$ is the virtual link between the $k^\text{th}$ user and the $n^\text{th}$ AP.

\begin{figure}
    \centering
    \includegraphics[width=.7\columnwidth]{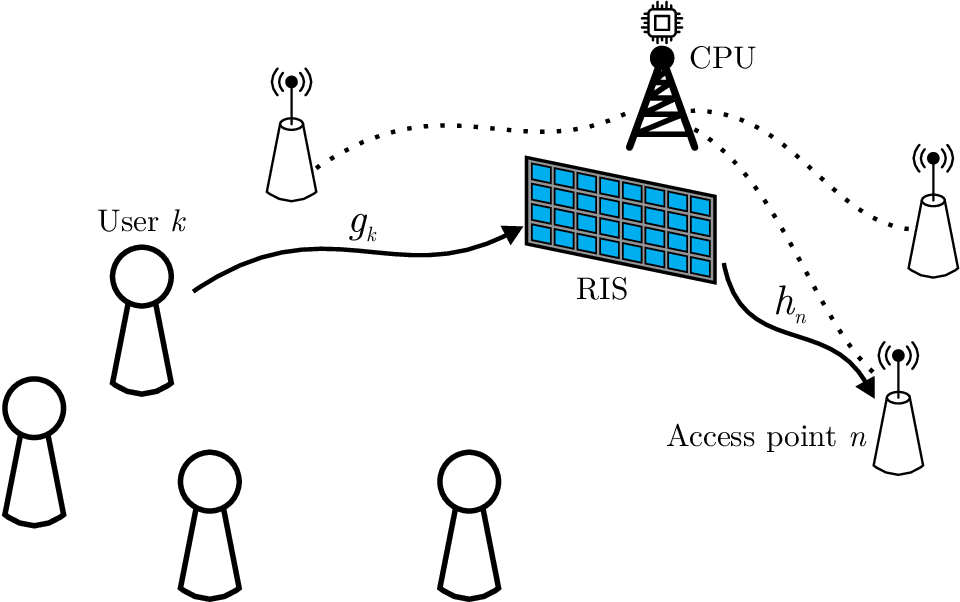}
    \caption{System Model}
    \label{fig:system}
\end{figure}

\par Similar models have been studied for massive MIMO systems in \cite{van2022reconfigurable}, and for cell-free MIMO in \cite{9665300}. For instance, the second moment of the aggregated channel can be directly obtained as:
\begin{align*}
\begin{split}
   & \E{\abs{u_{n,k}}^2} = \E{\abs{\bm{h}^\herm_n \, \bm{\Theta} \, \bm{g}_k}^2} = \E{\tr\left(\bm{h}_n^\herm \, \bm{\Theta} \, \bm{g}_k \, \bm{g}^\herm_k \, \bm{\Theta}^\herm
    \, \bm{h}_n\right)} \\
    &= \tr\left(\bm{\Theta}^\herm \, \E{\bm{h}_n \, \bm{h}_n^\herm} \, \bm{\Theta} \, \E{\bm{g}_k 
    \, \bm{g}^\herm_k}\right) = \tr\left(\bm{\Theta}^\herm \, \widetilde{\bm{R}}_n \, \bm{\Theta} \, \overline{\bm{R}}_k\right).
\end{split}
\end{align*}
We can also see that the channels between two different users and two different APs are mutually independent, $\forall n \neq n^\prime, k \neq k^\prime, \E{u_{n,k} \, u_{n^\prime,k^\prime}^*} = \E{u_{n,k}} \,\E{ u_{n^\prime,k^\prime}^*} = 0$. Moreover, the channels between a user and two different APs are uncorrelated, $\forall n \neq n^\prime, \E{u_{n,k} \, u_{n^\prime,k}^*} = \E{\left(\bm{h}_n^\herm \bm{\Theta} \bm{g}_k\right)\left(\bm{h}_{n^\prime}^\herm \bm{\Theta} \bm{g}_k\right)^*} = 0$, and the channels between two different users and the same AP are also uncorrelated, $\forall k \neq k^\prime, \E{u_{n,k} \, u_{n,k^\prime}^*} = \E{\left(\bm{h}_n^\herm \bm{\Theta} \bm{g}_k\right)\left(\bm{h}^\herm_n \bm{\Theta} \bm{g}_{k^\prime}\right)^*} = 0$. Finally, we define the aggregated channel covariance matrix $\bm{Q}_k = \E{\bm{u}_k \, \bm{u}_k^\herm} = \overline{\alpha}_k \tr\left(\bm{\Theta}^\herm \, \bm{R}_\text{\fontfamily{cmr}\selectfont t} \, \bm{\Theta} \, \bm{R}_\text{\fontfamily{cmr}\selectfont r}\right) \diag\left(\widetilde{\alpha}_1, \dots,\widetilde{\alpha}_N\right)$.

\subsection{UPLINK TRANSMISSION}

\par In a simultaneous manner, the $K$ users transmit their signals to the $N$ APs, that know an estimate $\widehat{u}_{n,k}$ of their channels, computed during the training phase. The baseband received signal at AP $n$ is:
\begin{equation}
    y_n = \sum_{k=1}^K u_{n,k} \, s_k + n_n
\end{equation}
where the additive noise at the $n^\text{th}$ AP $n_n \sim \mathcal{CN}\left(0,\, \sigma^2 \right)$ and the uplink transmitted symbol $s_k \in \C$ satisfies $\E{\abs{s_k}^2} = P_k$, with $P_k > 0$ being the transmit power budget of user $k$.
\par For data detection of the $k^\text{th}$ symbol, the $n^\text{th}$ AP multiplies the received signal $y_n$ with the conjugate of its locally obtained channel estimate. The acquired terms are then sent to the CPU via the backhaul network. Therefore, the decision statistic for user $k$'s symbol $s_k$ reads:
\begin{equation}\label{eq:4}
    r_k = \sum_{n=1}^N \widehat{u}^*_{n,k} \, y_n = \sum_{n=1}^N \sum_{j=1}^K \widehat{u}^*_{n,k}
    \, u_{n,j} s_j + \sum_{n=1}^N \widehat{u}^*_{n,k} \, n_n.
\end{equation}
In the following, to maintain analytical tractability, we assume that all users employ mutually orthogonal reverse link pilot sequences. Accordingly, channel estimation error can be neglected when the number of RIS elements $M$ is large, as shown in \cite{wang2021massive}. Hence, $\widehat{u}_{n,k} = u_{n,k}$.
Given the instantaneous signal to interference and noise ratio of user $k$ $\SINR_k$, his ergodic uplink rate is then $\E{\log\left(1+\SINR_k\right)}$. Due to the mathematical complexity of obtaining such closed-form expressions, we resort to a well-known tight lower bound achievable expression of the rate as follows. By utilizing the use-and-then-forget bounding technique as in  \cite[Theorem 4.4]{bjornson2017massive} and the received signal expression in \eqref{eq:4}, we can show that an achievable uplink rate expression for user $k$ is $\log\left(1+\ESINR_k\right)$, where $\ESINR$ is the effective uplink $SINR$ (which is of course different from average $SINR$), that can be computed by dividing the average  useful signal and average interference as follows (this rate expression $\log\left(1+\ESINR_k\right)$ is an extremely tight lower bound  as shown by empirical simulations in \mbox{\cite{9665300}}):
\begin{equation*}
    \ESINR_k = \frac{\abs{\mathtt{DS}_k}^2 P_k}{\sum\limits_{\substack{j = 1 \\ j \neq k}}^K \E{\abs{\mathtt{UI}_{j,k}}^2} P_j + \E{\abs{\mathtt{BU}_k}^2} P_k + \E{\abs{\mathtt{NO}_k}^2}}
\end{equation*}
where we have
\begin{equation*}\label{eq:6}
    \abs{\mathtt{DS}_k}^2 = \abs{\E{\sum\limits_{n=1}^N \widehat{u}^*_{n,k} u_{n,k}}}^2
    = \abs{\E{\bm{u}_k^\herm \bm{u}_k}}^2 = \left(\tr\left(\bm{Q}_k\right)\right)^2
\end{equation*}
\begin{align*}
\begin{split}
    \E{\abs{\mathtt{UI}_{j,k}}^2} &= \E{\abs{\sum\limits_{n=1}^N \widehat{u}^*_{n,k} \, u_{n,j}}^2}
    = \E{\abs{\bm{u}_k^\herm \bm{u}_j}^2} = \tr\left(\bm{Q}_k \, \bm{Q}_j\right)
\end{split}
\end{align*}
\begin{equation*}
    \E{\abs{\mathtt{NO}_k}^2} = \E{\abs{\sum\limits_{n=1}^N \widehat{u}^*_{n,k} \, n_n}^2}
    = \E{\abs{\bm{u}_k^\herm \bm{n}}^2} = \sigma^2 \tr\left(\bm{Q}_k\right)
\end{equation*}
\begin{align*}
\begin{split}
   & \E{\abs{\mathtt{BU}_k}^2} = \E{\abs{\sum\limits_{n=1}^N \widehat{u}^*_{n,k} \, u_{n,k} - 
    \E{\sum\limits_{n=1}^N \widehat{u}^*_{n,k} \, u_{n,k}}}^2} 
    \\ & = \E{\abs{\bm{u}_k^\herm\bm{u}_k}^2}-\abs{\E{\bm{u}_k^\herm\bm{u}_k}}^2
    = \mathbb{V}\left[\bm{u}_k^\herm \bm{u}_k\right] = \tr\left(\bm{Q}_k^2\right)
\end{split}
\end{align*}
where the final equality follows from \mbox{\cite[Lemma B.14]{bjornson2017massive}}  and \mbox{\eqref{eq:6}}, because the aggregated channels $\bm{u}_k$ can be approximated by Gaussian distributions when $M$ is large as shown in \cite{wang2021massive}. These variables represent the strength of the desired signal $\left(\mathtt{DS}_k\right)$, the beamforming gain uncertainty $\left(\mathtt{BU}_k\right)$, the interference caused by user $j$ on user $k$ $\left(\mathtt{UI}_{j,k}\right)$ and the power of additive noise $\left(\mathtt{NO}_k\right)$. After some algebraic manipulations, we simplify the $\ESINR$, denoted by $\hat{\gamma}$, expression to:
\begin{align} \label{eq:SINR}
    &\hat{\gamma}_k =  \frac{\overline{\alpha}_k^2 \left(\sum_{n=1}^N \widetilde{\alpha}_n\right)^2 \eta \, P_k}
    {\sum_{j=1}^K \overline{\alpha}_k \, \overline{\alpha}_j \left(\sum_{n=1}^N \widetilde{\alpha}_n^2\right) \eta \, P_j + \sigma^2 \, \overline{\alpha}_k \left(\sum_{n=1}^N \widetilde{\alpha}_n\right)}
\end{align}
where $\eta = \tr\left(\bm{\Theta}^\herm \, \bm{R}_\text{\fontfamily{cmr}\selectfont t} \, \bm{\Theta} \, \bm{R}_\text{\fontfamily{cmr}\selectfont r}\right)$. 


\subsection{USERS / FLOWS ARRIVAL MODEL}
\par All previous works on RIS-aided systems make the assumption that the users are static in the network, in the sense that a fixed number of devices constantly communicate with the BS or the APs. Nevertheless,  in real settings, users join the network dynamically, exchange bursts of data with the APs, and leave once they are served. Subsequently, we adopt the model used in \cite{assaad2018power} for power control in massive MIMO networks, and consider this dynamic user population around the RIS.
Specifically, we consider a finite, but possibly large, number of locations in the network in the vicinity of the RIS. Accordingly, we use $K$ to denote the number of locations instead of users. Let $X_k(t)$ represent the number of users at location $k$ at time $t$. Each user has one flow to be served, where a flow depicts a file that the user would like to convey to the APs. Thus, the words `flow' and `user' indicate a certain file transfer, and will be used interchangeably for the rest of this paper. Note that a similar model can be used for the downlink. Also, one can notice that a model where each user has multiple flows is a simple extension to this case. The arrival of the flows at location $k$ is modeled as a Poisson process with rate $\lambda_k$. These rates represent the average number of users arriving to each location. The size of the files to be transmitted is an exponentially distributed random variable $S_k$ at location $k$ with mean $\E{S_k}$. 
We consider distinct time-scales for the physical layer and the flow level, and we use the time index $t$ to refer to the time at the level of flows. In practice, multiple physical layer timeslots occur between times $t$ and $t+1$. In the rest of this paper, we will suppress the use of the time index when ambiguity is unlikely. Next, we define the network stability to be used in the stability analysis.
\begin{definition}[\textit{Strong Stability} \cite{neely2010stability}]
The network is said to be stable if $\lim\limits_{T \to \infty} \sup \,\frac{1}{T} \sum_{T=0}^{T-1} \sum_{k=1}^K \E{X_k(t)} < \infty$
\end{definition}

\par In simple terms, based on this stability definition, if the network is stable then all arriving users are expected to transmit their files in finite time. We let $\bm{\lambda} = \left[\lambda_1, \dots, \lambda_K\right]^\intercal$ denote the vector of user arrival rates. We can now present the definition of a stability region.
\begin{definition}
The stability region is the set of all mean arrival rate vectors for which there exists a RIS configuration that makes the network stable.
\end{definition}

\par We are now interested in finding a RIS configuration that stabilizes the network whenever possible.

\section{RIS PHASE SHIFTS DESIGN} \label{sec:3}
First, we formulate the problem we are interested in solving. Let $\bm{\Lambda}^\text{max}$ be the stability region under RIS configurations. Our target is then:
\begin{equation*}
    \begin{aligned}
        &\textrm{Choose} \quad && \bm{\theta} \in \left[0, 2\pi\right]^M\\
        &\textrm{such that} \quad && \forall \bm{\lambda} \in \bm{\Lambda}^\textrm{max},\, \textrm{the network is stable.}
    \end{aligned}
\end{equation*}

\par Note that the flows' arrival rates $\lambda_k$ might not be known a priori at the CPU; and hence, the network must be stabilized only by modifying the RIS phase shifts for any arrival rates lying inside the stability region. To accomplish this, we consider the following optimization problem:
\begin{equation*}
    \label{op:P1}
    \begin{aligned}
        &\textrm{(P1)}& \quad &\underset{\bm{\theta}}{\textrm{maximize}} \quad && 
        f\left(\bm{\theta}\right) = \sum_{k=1}^K X_k \, U_k\left(\bm{\theta}\right)\\
        &&&\textrm{subject to}&& \quad 0 \leq \theta_i \leq 2\pi \quad i=1, \dots, M\\
    \end{aligned}
\end{equation*}
where we choose the utility functions: $U_k\left(\bm{\theta}\right) = \log\left(ESINR_k\right)$\footnote{Throughout this paper, we use $\log\left(\cdot\right)$ as the natural logarithm function for simplicity, even in the bit rate expression. The analysis remains true for any other base, since all logarithms are equal up to a scaling factor.}. This choice is motivated by the concavity of the equivalent utility functions with respect to the allocated rates \cite{assaad2018power}. This asset will be discussed in the next section.
We will show in the next section that to stabilize the network for dynamic arrivals and departures of users, it is sufficient to solve the aforementioned problem \hyperref[op:P1]{(P1)}. In the remaining of this section, we will first show how to solve problem \hyperref[op:P1]{(P1)} and the stability analysis will be  provided in the next section.
\begin{remark}
We strongly emphasize that we do not make any kind of approximation in \hyperref[op:P1]{(P1)}, such as $\log\left(1 + \ESINR\right) \approx \log\left(\ESINR\right)$ for high $\ESINR$ values. The objective function can be seen as a sum of each location's utility functions, selected as $\log\left(\ESINR_k\right)$, weighted by the traffic volume at each location $X_k$. Regardless, the rate is still equal to $\log\left(1 + \ESINR\right)$. Informally, \hyperref[op:P1]{(P1)} can be seen as a variation to the proportional fairness problem but at the $\ESINR$ level.
\end{remark}

\par Solving this optimization problem is challenging, because of the intricate form of its objective function. To mitigate this difficulty, we start by showing that the objective function $f\left(\bm{\theta}\right)$ is increasing in $\eta$:
\begin{align}
\begin{split}
    &\frac{\partial f\left(\bm{\theta}\right)}{\partial \, \eta} = \sum_{k=1}^K X_k \times \\
    &\frac{\sigma^2 \, \overline{\alpha}_k \left(\sum_{n=1}^N \widetilde{\alpha}_n \right)}
    {\eta \left(\sigma^2 \, \overline{\alpha}_k \left(\sum_{n=1}^N \widetilde{\alpha}_n \right)
    + \sum_{j=1}^K \overline{\alpha}_k \, \overline{\alpha}_j \left(\sum_{n=1}^N \widetilde{\alpha}_n^2\right) \eta \, P_j\right)} 
    > 0.
\end{split}
\end{align}
On the other hand, notice that\footnote{This is a direct application of \cite[Theorem 1.11]{zhang2017matrix}.} $\eta = \tr\left(\bm{\Theta}^\herm \, \bm{R}_\text{\fontfamily{cmr} \selectfont t} \, \bm{\Theta} \, \bm{R}_\text{\fontfamily{cmr}\selectfont r}\right) = \bm{\phi}^\herm \, \bm{R} \, \bm{\phi}$, where $\bm{R} = \bm{R}_\text{\fontfamily{cmr} \selectfont t} \odot \bm{R}^\intercal_\text{\fontfamily{cmr}\selectfont r}$ is positive semi-definite by virtue of the Schur product theorem \cite{Schur1911}. Thus, \hyperref[op:P1]{(P1)} is equivalent to:
\begin{equation*}
    \label{op:P2}
    \begin{aligned}
        &\textrm{(P2)}& \quad &\underset{\bm{\phi}}{\textrm{maximize}} \quad && 
        \bm{\phi}^{\herm} \, \bm{R} \, \bm{\phi}\\
        &&&\textrm{subject to} \quad && \abs{\phi_i} = 1 && \quad i=1, \dots, M
    \end{aligned}
\end{equation*}
\hyperref[op:P2]{(P2)} is a complex quadratic optimization problem, that is in the class of NP-hard problems; and therefore a globally optimal solution cannot be obtained in polynomial time. However, we know in this case, that applying the semi-definite relaxation (SDR) technique reported in \cite{so2007approximating} can bound the error committed while approximating the solution. As such, the following semi-definite program (SDP) provides a relaxation for \hyperref[op:P2]{(P2)}:
\begin{equation*}
    \begin{aligned}
    \label{op:P3}
        &\textrm{(P3)}& \quad &\underset{\bm{\Phi}}{\textrm{maximize}} \quad && 
        \tr\left( \bm{R} \, \bm{\Phi}\right) \\
        &&&\textrm{subject to} \quad && \bm{\Phi} \succeq \bm{0},\; \bm{\Phi}_{i,i} = 1 && \quad i=1, \dots, M
    \end{aligned}
\end{equation*}
where $\bm{\Phi} = \bm{\phi}\, \bm{\phi}^\herm$. \hyperref[op:P3]{(P3)} is standard convex optimization problem that can be optimally solved by invoking interior-point based solvers, such as CVX \cite{grant2014cvx}. If the obtained solution $\bm{\Phi}^\star$ is rank-one, then it is also the optimal solution to \hyperref[op:P2]{(P2)}. But since this is not generally the case, we construct a rank-one solution to \hyperref[op:P2]{(P2)} from $\bm{\Phi}^\star$ using Guassian randomization as follows. We first compute the eigenvalue decomposition $\bm{\Phi}^\star = \bm{U}\,\bm{\Sigma}\,\bm{U}^\herm$, where  $\bm{U}$ and $\bm{\Sigma}$ are a unitary and a diagonal matrix respectively. Then we obtain a sub-optimal solution to \hyperref[op:P2]{(P2)} as $\overline{\bm{\phi}} = \bm{U} \, \bm{\Sigma}^{1/2} \, \bm{r}$ where $\bm{r} \in \C^{M \times 1}$ is a random vector generated according to $\bm{r} \sim \mathcal{CN}\left(\bm{0}_{M \times 1},\, \bm{I}_M\right)$. By drawing independent random vectors $\bm{r}$, the objective function of \hyperref[op:P2]{(P2)} is approximated by the maximum one attained by the corresponding $\overline{\bm{\phi}}$. We finally recover a suboptimal solution to \hyperref[op:P2]{(P2)} as a vector $\bm{\phi}$ with elements $\phi_i = e^{j \arg\left(\overline{\phi}_i\right)}, 1 \leq i \leq M$. It has been shown in \cite{so2007approximating},~\cite{luo2010semidefinite} that such techniques with a sufficiently large number of randomizations, guarantee a $\frac{\pi}{4}$-approximation to \hyperref[op:P2]{(P2)}, in following sense:
\begin{equation}\label{eq:imp-aprox}
    \frac{\pi}{4} \, \eta^{\text{opt}} \leq \eta^{\text{sub}} \leq \eta^{\text{opt}}
    \quad \text{or} \quad \frac{\pi}{4} \leq \gamma = \frac{\eta^{\text{sub}}}{\eta^{\text{opt}}} \leq 1
\end{equation}
where $\eta^{\text{opt}}$ and $\eta^{\text{sub}}$ are the optimal and approximate solutions of \hyperref[op:P2]{(P2)}, and $\gamma$ is the approximation accuracy.

\par It is worth noting that a similar method can be applied to solve the same problem in the case the RIS phase shifts can only take a finite number of discrete values. In such scenarios, $\bm{\phi} \in \mathcal{V}_L = \{1, \omega, \dots, \omega^{L-1}\}$, where $\omega$ is the principal $L^\text{th}$ root of unity and $L$ is the number of phase shift levels. The minimum approximation accuracy becomes $\frac{\left(L \sin\left(\frac{\pi}{L}\right)\right)^2}{4 \pi}$.

\section{STABILITY ANALTSIS OF THE NETWORK} \label{sec:4}
\par In this part, we will show that the RIS configuration presented in the previous section, i.e. the one that solves \hyperref[op:P1]{(P1)}, guarantees the network stability for any flow arrival rates inside the stability region $\bm{\Lambda}^\text{max}$ when \hyperref[op:P1]{(P1)} is solved with no optimality gap, and for all flow arrival rates that are within a factor proportional to the optimality gap otherwise. To do so, we utilize fluid limit analysis to demonstrate the stability. In other terms, for the stochastic process $\mathbf{X}(t) = \left[X_1(t), \dots, X_K(t)\right]^\intercal$ that portrays the flows' volumes, we introduce a deterministic process $\mathbf{Y}(t)$ that approximates the evolution of $\mathbf{X}$ subject to a certain limiting criteria. It is known that if the fluid limit associated to a stochastic process reaches zero in finite time, then the process itself is stable. Furthermore, to show that $\mathbf{Y}$ is stable, it is sufficient to prove that its corresponding Lyapunov function exhibits negative drift under the selected RIS configuration \cite{bonald2001impact}, \cite{andrews2004scheduling}.

\par Recall that the users' arrivals follow a Poisson distribution with mean $\lambda_k$ at location $k$. We suppose that when a new user joins the network, he directly starts a connection with the APs and transmits a file having an average size of $\E{S_k}=1 \, \forall \, k$. It will then be straightforward to extend this situation to different file size means and renewal arrival processes. In this context, it is clear that $\mathbf{X}$ is a Markov process that has the following evolution at each timeslot $t$, and at each location $k$:
\begin{align*}
    &&X_k \longrightarrow X_k + 1 &&\text{at rate} &&\lambda_k \\
    &&X_k \longrightarrow X_k - 1 &&\text{at rate} &&R_k
\end{align*}
where $R_k$ is the physical layer rate allocated to location $k$ between time instants $t$ and $t+1$. This rate is expressed as $R_k = \log\left(1 + \ESINR_k\right)$, where the $\ESINR$ at each location is given by \eqref{eq:SINR}.
\par After introducing all the necessary ingredients, we now state the main result of this paper.
\begin{theorem}
By choosing the RIS phase shifts that solve \hyperref[op:P1]{(P1)}, the network is stable for any user arrival rate vector $\gamma \bm{\lambda}$, where $\bm{\lambda} \in \bm{\Lambda}^\text{max}$ and $\gamma \in \left[\frac{\pi}{4}, 1\right]$ is the approximation accuracy of \hyperref[op:P2]{(P2)}.
\end{theorem}
The aforementioned theorem implies that our proposed solution allows achieving a guaranteed fraction higher than $\pi/4$ (i.e. higher than 78.5\%) of the stability region. 
\begin{proof}

Before proving the theorem, we provide a useful lemma for the demonstration.
\begin{lemma} \label{lemma1}
$\forall\, \gamma \in (0, 1], \forall\, x > 0$, $\log \left( 1 + \gamma x \right) \geq \gamma \log \left( 1 + x \right)$, with equality when $\gamma = 1$.
\end{lemma}
\begin{proof}
When $\gamma=1$ the equality is obvious. Otherwise, $\forall\, \gamma \in (0, 1), \forall\, x > 0$, the proof follows by composing both sides of Bernoulli's inequality \cite{bullen2013handbook} $\left(1+x\right)^\gamma < 1+\gamma x$ with the increasing function $x \mapsto \log(x)$.
\end{proof}
The proof consists of studying the fluid system obtained when the initial number of flows grows to infinity. Specifically, we consider the set of fluid limits defined by:
\begin{equation*}
    Y_k(t) = \lim_{\beta \to \infty} \frac{X_k\left(\beta t\right)}{\beta}
    \quad\text{with}\quad \sum_{k=1}^K X_k(0) = \beta.
\end{equation*}
Notice that if the limit exists, $\sum\limits_{k=1}^K Y_k(0) = 1$. Given this initial distribution of the fluid system $\mathbf{Y}(0)$, the evolution of $\mathbf{Y}(t)$ is uniquely defined by the following set of differential equations, given by the strong law of large numbers:
\begin{equation} \label{eq:differential}
    \dv{t} Y_k(t) = \lambda_k - R_k \quad\text{for all}\; k, t\;\text{such that}\, Y_k(t) > 0.
\end{equation}
Now, we define $\mathbf{R} = \left[R_1, \dots, R_K\right]^\intercal$ as a vector containing the allocated rates at each location. Recall that the rate is calculated by $R_k = \log\left(1+\ESINR_k\right)$. Conversely we can write $\log\left(\ESINR_k\right) = \log\left(e^{R_k} - 1\right)$. Then, \hyperref[op:P1]{(P1)} is equivalent to the following optimization problem:
\begin{equation*}
    \label{op:P4}
    \begin{aligned}
        &\textrm{(P4)}& \quad &\underset{\mathbf{R}}{\textrm{maximize}} \quad && 
        \sum_{k=1}^K Y_k \log\left(e^{R_k}-1\right)\\
        &&&\textrm{subject to} \quad && 0 \leq f_i\left(\mathbf{R}\right) \leq 2\pi, \quad i=1, \dots, M
    \end{aligned}
\end{equation*}
where $\mathbf{X}$ is interchanged with its limit $\mathbf{Y}$, and the constraint functions $f_i\left(\mathbf{R}\right)$ are the equivalent of the phase shift constraints on $\bm{\theta}$. Let $\ESINR_k\left(\eta\right)$ denote the $\ESINR$ value at location $k$ for a given value of $\eta$. Since $\ESINR_k\left(\eta\right)$ increases with $\eta$, we have from \eqref{eq:imp-aprox}:
\begin{align}
\begin{split}\label{ineq:14}
    \gamma\,\ESINR_k\left(\eta^\text{opt}\right) \overset{\text{(a)}}{\leq} \ESINR_k\left(\gamma\,\eta^\text{opt}\right) &\leq
    \ESINR_k\left(\eta^\text{sub}\right) \\ &\leq  \ESINR_k\left(\eta^\text{opt}\right)
\end{split}
\end{align}
where (a) is a straightforward bound due to $0 < \gamma \leq 1$. Moreover, for any $x>0$ the function $x \mapsto \log\left(1+x\right)$ is increasing and thus,
\begin{align}
\begin{split}\label{ineq:15}
    \gamma\,\log\left(1 + \ESINR_k\left(\eta^\text{opt}\right)\right) \overset{\text{(b)}}{\leq} \log\left(1 + \gamma\,\ESINR_k\left(\eta^\text{opt}\right)\right) \\ \leq
    \underbrace{\log\left(1 + \ESINR_k\left(\eta^\text{sub}\right)\right)}_{R_k^\text{sub}} \leq \underbrace{\log\left(1 + \ESINR_k\left(\eta^\text{opt}\right)\right)}_{R_k^\text{opt}}
\end{split}
\end{align}
where (b) is a direct application of Lemma \ref{lemma1}. We define $\mathbf{R}^\text{sub} = \left[R^\text{sub}_1, \dots, R^\text{sub}_K\right]^\intercal$ and $\mathbf{R}^\text{opt} = \left[R^\text{opt}_1, \dots, R^\text{opt}_K\right]^\intercal$. Due to the equivalence between \hyperref[op:P1]{(P1)} and \hyperref[op:P4]{(P4)}, $\mathbf{R}^\text{sub}$ and $\mathbf{R}^\text{opt}$ are, respectively, the approximate and optimal solutions to \hyperref[op:P4]{(P4)}. The series of inequalities in \eqref{ineq:14} and \eqref{ineq:15} mean that if $\gamma$ is the approximation accuracy of \hyperref[op:P2]{(P2)} in the sense of \eqref{eq:imp-aprox}, then $\gamma$ is also an approximation accuracy in the same sense for the $\ESINR$ and the rate at all locations. 

On the other hand, the objective function of \hyperref[op:P4]{(P4)} is the sum of strictly concave functions $\log\left(e^{R_k}-1\right)$ with respect to $R_k$. These concave functions are always upper bounded by the first order of their Taylor expansion. Particularly, we have:
\begin{equation} \label{eq:14}
    \log\left(e^{R^\text{sub}_k}-1\right) \leq
    \log\left(e^{\gamma \lambda_k}-1\right) + \frac{e^{\gamma \lambda_k}}{e^{\gamma \lambda_k}-1}\left(R^\text{sub}_k- \gamma \lambda_k\right)
\end{equation}
for any arrival rate $\lambda_k$. Furthermore, for $x>0$, the function $x \mapsto \log\left(e^x-1\right)$ is strictly increasing. By (\ref{ineq:15}), we have  $\sum_{k=1}^K Y_k \log\left(e^{\gamma R^\text{opt}_k}-1\right) \leq  \sum_{k=1}^K Y_k \log\left(e^{ R^\text{sub}_k}-1\right)$. 

On the other hand, one can show easily that  $\eta^\text{opt}$ is also the optimal solution to the optimization problem $\sum_{k=1}^K X_k log((1+ESINR)^\gamma-1)$, by proving that the objective function is monotone with respect to $\eta$ as done previously in this paper. Therefore, $R^\text{opt}$ is also optimal for $\sum_{k=1}^K Y_k \log\left(e^{\gamma R_k}-1\right) $. This implies that $\sum_{k=1}^K Y_k \log\left(e^{\gamma R^\text{opt}_k}-1\right) \geq  \sum_{k=1}^K Y_k \log\left(e^{ \gamma \lambda_k}-1\right)$. 

Consequently,  for any arrival rate vector $\bm{\lambda}$ inside the stability region, we can write:
\begin{equation}\label{eq:16}
    \sum_{k=1}^K Y_k \frac{e^{ \gamma \lambda_k}}{e^{ \gamma \lambda_k}-1}\left( \gamma \lambda_k - R^\text{sub}_k\right)
    \leq 0.
\end{equation}

We then introduce the quadratic Lyapunov function $\mathcal{L}\left(\mathbf{Y}(t)\right) = \sum_{k=1}^K \frac{1}{2} \frac{e^{\gamma \lambda_k + \epsilon}}{e^{\gamma \lambda_k + \epsilon}-1} Y^2_k(t)$, where for any arrival rate $\gamma \lambda_k$ lying strictly inside the stability region, we select $\epsilon>0$ such that $\forall k$, $\gamma \lambda_k + \epsilon$ is inside or on the boundary of the region. We get 
\begin{equation}\label{eq:17}
    \dv{t}\mathcal{L}\left(\mathbf{Y}\right) = 
    \sum_{k=1}^K Y_k \frac{e^{\gamma \lambda_k + \epsilon}}{e^{\gamma \lambda_k + \epsilon}-1}\left(\gamma \lambda_k+ \epsilon- R^\text{sub}_k\right) \leq 0
\end{equation}
In view of \eqref{eq:17}, this reads:
\begin{equation}
    \dv{t}\mathcal{L}\left(\mathbf{Y}\right) \leq - \epsilon
    \sum_{k=1}^K Y_k \frac{e^{\gamma \lambda_k + \epsilon}}{e^{\gamma \lambda_k + \epsilon}-1}.
\end{equation}
By using  the inequality $$2 \times \mathcal{L}\left(\mathbf{Y}\right)= \sum_{k=1}^K  \frac{e^{\gamma \lambda_k + \epsilon}}{e^{\gamma \lambda_k + \epsilon}-1} Y^2_k(t) \leq  \left( \sum_{k=1}^K  \frac{e^{\gamma \lambda_k + \epsilon}}{e^{\gamma \lambda_k + \epsilon}-1} Y_k(t)\right)^2,$$
we conclude that there exists a constant $\xi>0$ such that:
\begin{equation}\label{eq:19}
    \dv{t}\mathcal{L}\left(\mathbf{Y}\right) \leq - \xi \sqrt{\mathcal{L}\left(\mathbf{Y}\right)}
\end{equation}
This implies for all $t \geq 0$ such that $\mathcal{L}\left(\mathbf{Y}(t)\right) \geq 0$:
\begin{equation}\label{eq:20}
    \mathcal{L}\left(\mathbf{Y}(t)\right) \leq 
    \left(\sqrt{\mathcal{L}\left(\mathbf{Y}(0)\right)} - \frac{\xi}{2} \, t\right)^2
\end{equation}
and also $\sqrt{\mathcal{L}\left(\mathbf{Y}(t)\right)} \leq 
    \left(\sqrt{\mathcal{L}\left(\mathbf{Y}(0)\right)} - \frac{\xi}{2} \, t\right)$ since $\mathcal{L}\left(\mathbf{Y}(t)\right) \geq 0$. Therefore, the inequality in \eqref{eq:19} means that if there exists some $T>0$ with $\mathcal{L}\left(\mathbf{Y}(T)\right)=0$, then $\mathcal{L}\left(\mathbf{Y}(t)\right)=0$ for all $t \geq T$. Looking at \eqref{eq:20}, we select:
\begin{equation}
    T = \frac{2}{\xi} \sqrt{\sum_{k=1}^K \frac{1}{2} \frac{e^{\gamma \lambda_k + \epsilon}}{e^{\gamma \lambda_k + \epsilon}-1}}
\end{equation}
This implies that for all $t \geq T$, $\mathcal{L}\left(\mathbf{Y}(t)\right)$ and thus $\mathbf{Y}(t)$ are identically zero, and the system is stable.
\end{proof}

\par We provide a summary of the proof. First, we represent the flow model by a vector $\mathbf{X}(t)$ containing the number of users at each location. Given the arrival model, we obtain the temporal equation describing the evolution of $\mathbf{X}$. Then, using the law of large numbers, a similar limit equation can be obtained for its fluid limit $\mathbf{Y}$. Knowing that the allocated rate at each location should be strictly larger than the arrival rate in order to guarantee the stability, we bound the suboptimal rates obtained from \hyperref[op:P1]{(P1)} and show that they are optimal in the scaled stability region. Finally, we use the Lyapunov technique to show that $\mathbf{Y}$ reaches zero in finite time for these arrival rates. It is then a direct conclusion that the network is stable.

\begin{figure}
    \centering
    \includegraphics[width=0.9\columnwidth]{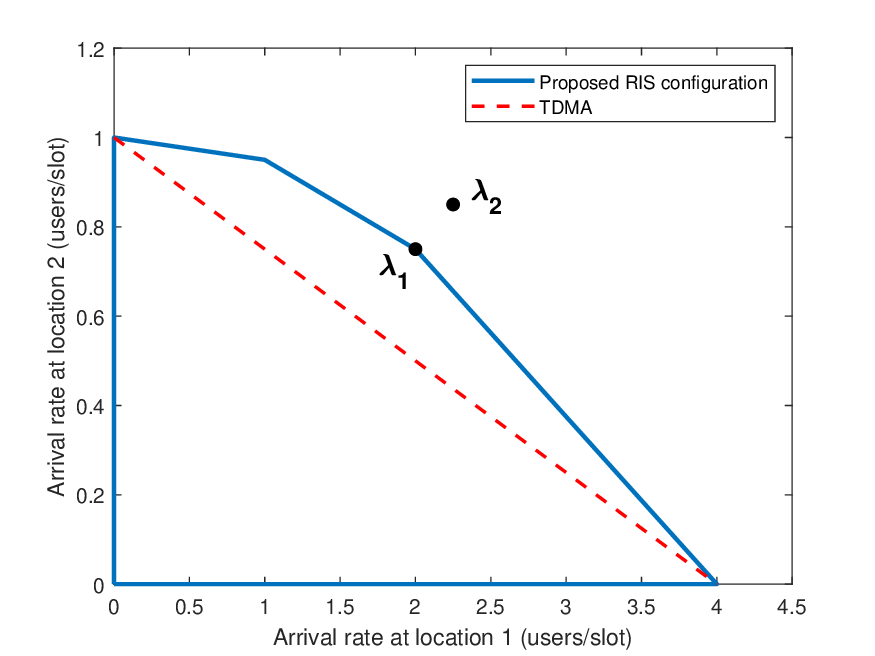}
    \caption{Simulation Results: Stability Region}
    \label{fig:sim1}
\end{figure}


\begin{figure}
    \centering
    \includegraphics[width=0.9\columnwidth]{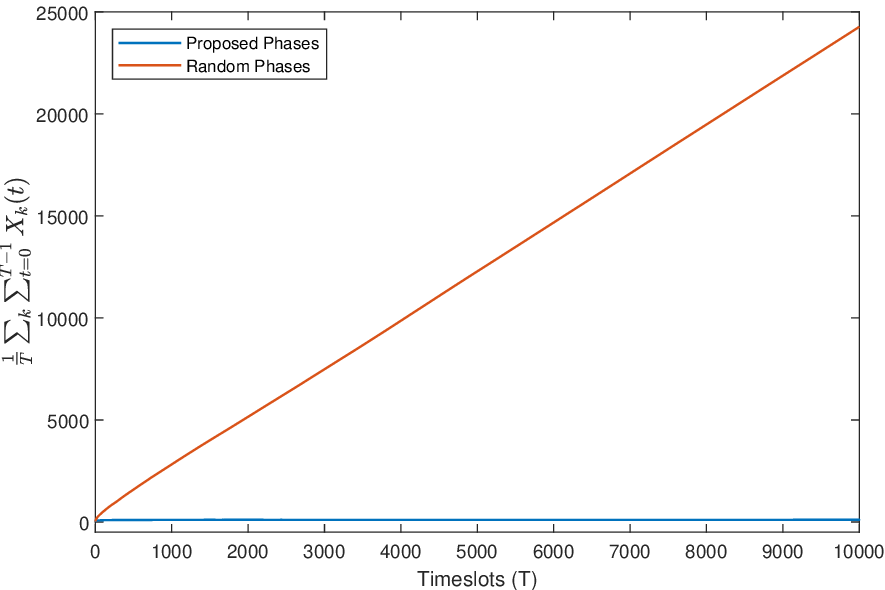}
    \caption{Simulation Results: Evolution of the stability metric for different phase shift designs. The arrival rate is fixed at 0.1 users/slot at all locations.}
    \label{fig:sim3}
\end{figure}

\section{NUMERICAL RESULTS} \label{sec:5}
In this section, we report some simulation results in order to illustrate the stability of the considered wireless system. We consider a $2 \times 2$ km$^2$ area with $N=128$ APs uniformly distributed in the lower left sub-region with coordinates $x,y \in [-0.75,-1]$ and the RIS with $M=1600$ elements is located at the origin. The bandwidth of the system is $20$ MHz and the carrier frequency is $1.9$ GHz. The large scale fading coefficients are generated according to the three-slope model in \cite[Section VI]{ngo2017cell}.

For visual considerations, we start by considering $K=2$ locations in the upper right region, with both coordinates being 0.25 and 0.75 for each location respectively. The users arrive to each location according to a Poisson process and transmit packets of $1$ Mbits with a power budget of $20$ dBm. In \autoref{fig:sim1} we compare the network stability region obtained by employing our proposed phase shifts design and the one obtained by using time division multiple access (TDMA) where at each timeslot, the RIS phases are optimized to maximize the $\SNR$ of the active user. We notice that the TDMA stability region is contained within the one guaranteed by our proposed RIS configuration. 

We now consider $K=400$ locations whose coordinates $x,y \in [0.05,1]$ form a square mesh in the upper right sub-region. In \autoref{fig:sim3}, we compare, for an arrival rate equal to $0.1$ users/slot at all locations, the sum over the locations of the moving average of the users density when using our proposed RIS configuration, to the one obtained by drawing each phase uniformly at random from $[0, 2\pi]$ at each timeslot. It is clear that the network is not stable when random RIS phases are used since the number of flows grows rapidly with time, while this number is bounded when using our proposed phase shifts. In \autoref{fig:sim4}, we plot the stability metric computed at $T=10000$ timeslots versus the arrival rate at all locations, for our proposed RIS phase shifts and the random phases scheme. We first observe that the metric varies exponentially with respect to the arrival rate. Secondly, we evince that the network is stable for a wider region of arrival rates when the RIS phases are optimized as we proposed, compared to the case where they are randomly selected, since the stability metric diverges at a prior point in the latter case.

\begin{figure}
    \centering
    \includegraphics[width=0.9\columnwidth]{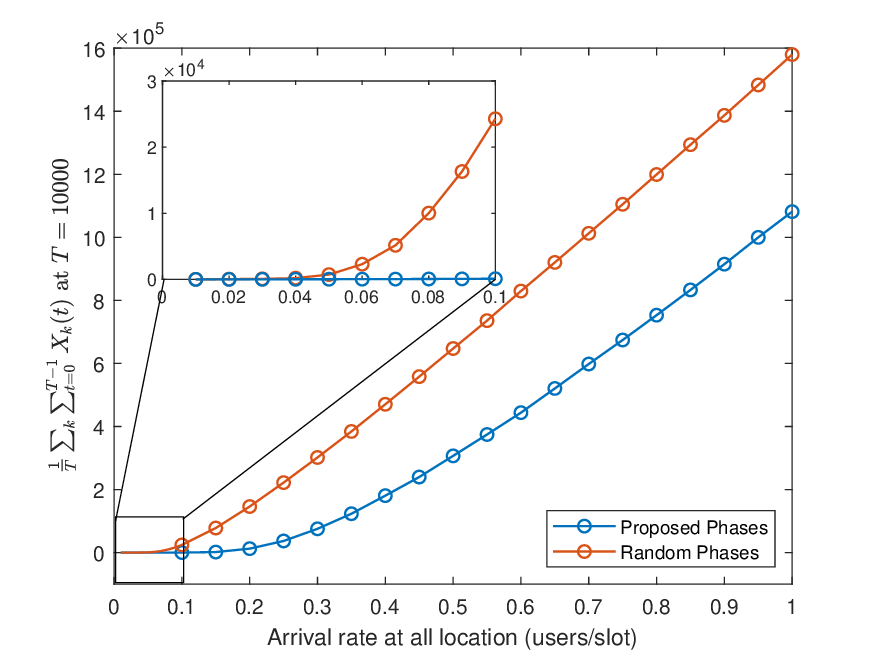}
    \caption{Simulation Results: Stability metric versus the arrival rate at all locations for different phase shift designs. The metric is calculated at T=10000 timeslots.}
    \label{fig:sim4}
\end{figure}

\section{CONCLUSION} \label{sec:6}
In this paper, we studied the network stability of a RIS-assisted cell-free MIMO system with dynamic arrivals and departures of users. We started by obtaining a closed form physical layer rate expression, and then described the data layer flow arrivals. Next, we proposed an optimization framework for the RIS induced phases and provided a low complexity sub-optimal solution of the optimization framework. We then proved rigorously that this solution achieves a high fraction of the stability region, and provided a lower bound of this fraction. 

\bibliographystyle{IEEEtran}
\bibliography{References.bib}

\end{document}